\documentclass[aps,pra,superscriptaddress,reprint,amsmath,amssymb,a4paper]{revtex4-1}
\usepackage{graphicx}
\usepackage{dcolumn}
\usepackage{bm}
\usepackage{color}
\usepackage{amssymb, amsmath, amsthm}
\makeatletter
\def\amsbb{\use@mathgroup \M@U \symAMSb}
\makeatother
\usepackage{textcomp}
\usepackage[T1]{fontenc}
\usepackage[mathscr]{euscript}
\usepackage{tikz}
\usetikzlibrary{arrows, calc, shapes.callouts, shapes.arrows}
\tikzset{>=stealth}
\newcommand{\hordots}[4]
{
	\draw [fill, color=#4] (#1 + 0.25, #2) circle [radius=#3];
	\draw [fill, color=#4] (#1, #2) circle [radius=#3];
	\draw [fill, color=#4] (#1 - 0.25, #2) circle [radius=#3];
}
\definecolor{darkred}{RGB}{150, 0, 0}
\definecolor{darkblue}{RGB}{0, 0, 150}

\newcommand{\nbox}[2][9]{\hspace{#1pt} \mbox{#2} \hspace{#1pt}}

\newcommand{\sA}{\mathscr{A}}
\newcommand{\sB}{\mathscr{B}}
\newcommand{\cS}{\mathcal{S}}
\newcommand{\cP}{\mathcal{P}}
\newcommand{\cX}{\mathcal{X}}
\newcommand{\cY}{\mathcal{Y}}
\newcommand{\cZ}{\mathcal{Z}}
\newcommand{\ave}[1]{\langle #1 \rangle}
\renewcommand{\th}
{
^{\textnormal{th}}
}
\newtheorem{lem}{Lemma}[section]
\newtheorem{prop}{Proposition}[section]

\begin{document}

\title{Practical relativistic bit commitment}

\author{T.~Lunghi}
\affiliation{Group of Applied Physics, University of Geneva, Chemin de Pinchat 22, CH-1211 Gen\`eve 4, Switzerland}
\author{J.~Kaniewski}
\affiliation{Centre for Quantum Technologies, National University of Singapore, 3 Science Drive 2, Singapore 117543}
\affiliation{QuTech, Delft University of Technology, Lorentzweg 1, 2628 CJ Delft, Netherlands}
\author{F.~Bussi\`{e}res}
\affiliation{Group of Applied Physics, University of Geneva, Chemin de Pinchat 22, CH-1211 Gen\`eve 4, Switzerland}
\author{R.~Houlmann}
\affiliation{Group of Applied Physics, University of Geneva, Chemin de Pinchat 22, CH-1211 Gen\`eve 4, Switzerland}
\author{M.~Tomamichel}
\affiliation{Centre for Quantum Technologies, National University of Singapore, 3 Science Drive 2, Singapore 117543}
\affiliation{School of Physics, The University of Sydney, Sydney 2006, Australia}
\author{S.~Wehner}
\affiliation{Centre for Quantum Technologies, National University of Singapore, 3 Science Drive 2, Singapore 117543}
\affiliation{QuTech, Delft University of Technology, Lorentzweg 1, 2628 CJ Delft, Netherlands}
\author{H.~Zbinden}
\affiliation{Group of Applied Physics, University of Geneva, Chemin de Pinchat 22, CH-1211 Gen\`eve 4, Switzerland}
\date{\today}

\begin{abstract}
Bit commitment is a fundamental cryptographic primitive in which Alice wishes to commit a secret bit to Bob. Perfectly secure bit commitment between two mistrustful parties is impossible through asynchronous exchange of quantum information. Perfect security is however possible when Alice and Bob each split into several agents exchanging classical information at times and locations suitably chosen to satisfy specific relativistic constraints. In this Letter we first revisit a previously proposed scheme~\cite{crepeau11} that realizes bit commitment using only classical communication. We prove that the protocol is secure against quantum adversaries for a duration limited by the light-speed communication time between the locations of the agents. We then propose a novel multi-round scheme based on finite-field arithmetic that extends the commitment time beyond this limit, and we prove its security against classical attacks.
Finally, we present an implementation of these protocols using dedicated hardware and we demonstrate a 2 ms-long bit commitment over a distance of 131~km. By positioning the agents on antipodal points on the surface of the Earth, the commitment time could possibly be extended to $212$~ms.
\end{abstract}

\maketitle

Bit commitment is a fundamental primitive with several applications such as coin tossing~\cite{Blum1981a}, secure voting~\cite{Broadbent2008a}, contract signing or honesty-preserving auctions~\cite{Boneh2000a}. In a bit commitment protocol, Alice commits a secret bit to Bob which she can choose to reveal some time later. Security here means that if Alice is honest, then her bit is perfectly concealed from Bob until she decides to open the commitment and reveal her bit. Furthermore, if Bob is honest, then it should be impossible for Alice to change her mind once the commitment is made. That is, the only bit she can unveil is the one she originally committed herself to. Information-theoretically secure bit commitment in a setting where the two mistrustful parties exchange classical messages in an asynchronous fashion is impossible.  An extensive amount of work was devoted to study asynchronous quantum bit commitment, for which perfect security was ultimately shown to be impossible~\cite{Mayers1997a,Lo1997a,Dariano2007a,winkler2011a}. 
Note however that arbitrarily long commitments are possible if one makes the assumption that the quantum memory of the dishonest party is bounded~\cite{damgard05, damgard07} or noisy~\cite{wehner08a, konig12}.

Alternatively, bit commitment with split agents exchanging classical information was proposed as early as 1988~\cite{benor:bc}. Security against classical attacks was proved under the condition that no communication was possible between some of the agents. This protocol was later simplified~\cite{crepeau11}, and the new scheme called simplified-\texttt{BGKW}, \texttt{sBGKW}~\cite{benor:bc} was proven secure against classical and a restricted class of quantum attacks. The possibility of enforcing the no-communication condition using relativistic constraints on the timing of the classical communication was formulated in~\cite{Kent1999}. This later led to the proposal of relativistic protocols based on the exchange of quantum and classical information~\cite{Kent2011, Kent2012a}, which were proved to be secure against quantum adversaries~\cite{Kaniewski2013a,Croke2012a}. Such protocols were experimentally demonstrated recently~\cite{Lunghi2013a,Liu2014a}. However, the commitment time achievable using these protocols is fundamentally bounded by half the time required to send light signals between the remote agents, i.e.~at most $\sim 21$~ms if they are constrained to be on the surface of the Earth. 

The possibility of extending the commitment to an arbitrary duration was proposed in 1999~\cite{Kent1999}. It relies on positioning one agent of Alice $\sA_{1}$ near an agent of Bob $\sB_{1}$ at an agreed upon location, and similarly agents $\sA_{2}$ and $\sB_{2}$ at another location. 
Carefully timed classical communication between $\sA_{i}$ and $\sB_{i}$ allows Alice to commit to a bit that she later reveals at a time of her choosing. This requires several rounds of communication, and the amount of communication increases exponentially with the number of rounds making it impractical. This limitation was later mitigated, at least in principle, using a compression scheme that requires only a constant communication rate~\cite{Kent2005}. Security argument against classical adversaries presented in Ref.~\cite{Kent2005} is of asymptotic nature and, therefore, not sufficient for implementation purposes.

In this Letter, we first revisit the \texttt{sBGKW} bit commitment protocol~\cite{crepeau11} that uses classical communication only.
We show that successful cheating  is equivalent to winning a non-local game analyzed in Ref.~\cite{sikora14}, thereby proving the security of this protocol against quantum adversaries. To the best of our knowledge, this is the first entirely classical protocol to be proven secure against arbitrary quantum adversaries. To extend the duration of the commitment beyond the communication time between the locations of the agents (which constitutes the relativistic constraint in the \texttt{sBGKW} scheme), we introduce a novel multi-round scheme based on finite-field arithmetic and we prove its security against classical adversaries. Our scheme is simple and efficient and the security argument leads to a natural, algebraic problem for which we prove explicit and quantitative bounds (see Proposition B.2 in the Supplemental Material (SM)). Finally, we present practical implementations of both the \texttt{sBGKW} scheme and the multi-round variant, and show how this could be used to realize commitments of duration reaching up to $\sim 212$ milliseconds.

\paragraph*{Security definition}
We take $n \in \mathbb{N}$ to be the security parameter and we interpret $n$-bit strings as elements of the finite field $\mathbb{F}_{2^{n}}$ (for compactness we write $0$ to denote $0^{n} = 00 \ldots 0$). We denote addition by ``$\oplus$'' (in this case it is just the bitwise XOR) and multiplication by ``$*$''. Moreover, if $d$ is a bit and $b$ is an $n$-bit string then we define
\begin{equation*}
d \cdot b =
\begin{cases}
0 &\nbox[6]{if} d = 0,\\
b &\nbox[6]{if} d = 1.
\end{cases}
\end{equation*}
All secret strings used in the protocol are chosen uniformly at random from $\{0, 1\}^{n}$.

Let Alice (who makes the commitment) and Bob (who receives the commitment) have agents at two distinct locations ($\sA_{1}$ and $\sB_{1}$ at Location 1; $\sA_{2}$ and $\sB_{2}$ at Location~2) and let $d \in \{0, 1\}$ be the bit that honest Alice wants to commit to. The protocol consists of multiple rounds which alternate between the two locations and the timing is chosen such that every two consecutive rounds are space-like separated. Hence, no message sent during a certain round from one location can reach the other location in time for the next round.

Security for honest Alice is quantified by Bob's ability to guess her commitment \emph{immediately before} the open phase (assuming he might deviate arbitrarily from the honest protocol). All the protocols considered in this paper are~\emph{perfectly hiding}, which means that Bob remains completely ignorant about Alice's commitment (his guessing probability equals $\frac{1}{2}$).

Security for honest Bob is quantified through a scenario in which Alice performs an arbitrary action in the commit phase and is \emph{immediately after} challenged to open one of the bits. Given a particular strategy adopted by Alice in the commit phase we define $p_{d}$ to be the optimal probability of successfully unveiling~$d$. The protocol is $\varepsilon$-\textit{binding} if
\begin{equation*}
p_{0} + p_{1} \leq 1 + \varepsilon
\end{equation*}
for all strategies of dishonest Alice in the commit phase. Note that this is a weak, non-composable definition of security. In Appendix C we discuss how to formalize these definitions in the relativistic setting. (For a general overview see Ref.~\cite{Kaniewski2013a}.)
\paragraph*{Security of the \texttt{sBGKW} scheme}
We now present the scheme proposed in Ref.~\cite{crepeau11} and we prove its security against quantum adversaries. Before the protocol begins $\sA_{1}$ and $\sA_{2}$ must share a secret $n$-bit string $a$. Note that $\sB_{1}$ also needs a secret string $b$ but it can be generated before or during the protocol. The protocol consists of two rounds:
\begin{enumerate}
\item (commit) $\sB_{1}$ sends $b$ to $\sA_{1}$. $\sA_{1}$ returns $(d \cdot b) \oplus a$ to $\sB_{1}$.
\item (open) $\sA_{1}$ unveils the committed bit $d$ to $\sB_{1}$ while $\sA_{2}$ sends $a$ to $\sB_{2}$.
\end{enumerate}
To check whether the commitment should be accepted $\sB_{1}$ and $\sB_{2}$ need to communicate (e.g.~through an authenticated channel) and verify that the string returned by $\sA_{1}$ in the commit phase equals $(d \cdot b) \oplus a$.

Security for honest Alice comes from the fact that the only message that Bob receives in the commit phase is a uniformly random string.

Security for honest Bob in the classical case is fairly intuitive: in order for $\sA_{2}$ to be able to unveil both commitments she would need to know both $a$ and $a \oplus b$, hence, she would know $b$. However, since $b$ is chosen uniformly at random by Bob this must be difficult. This argument can be made rigorous~\cite{crepeau11} to show that the protocol is $\varepsilon$-binding for $\varepsilon = 2^{-n}$ (and this is actually tight: the trivial strategy of always outputting $0$ gives $p_{0} = 1$ and $p_{1} = 2^{-n}$). Unfortunately, this reasoning does not work against quantum adversaries since $\sA_{2}$ could have two distinct measurements that reveal $a$ and $a \oplus b$, respectively, but since they could be incompatible this would not have direct implications on her ability to guess $b$.

To find an explicit bound on $p_{0} + p_{1}$ we formulate cheating as a non-local game in which $\sA_{1}$ receives $b$, $\sA_{2}$ receives $d$ (the bit she is required to unveil) and the XOR of their outputs is supposed to equal $d \cdot b$. Winning such a game with probability $p_{\textnormal{win}}$ corresponds to a cheating strategy that achieves $p_{0} + p_{1} = 2 p_{\textnormal{win}}$. More concisely, the rules of the non-local game are~\cite{crepeau11}:
\begin{enumerate}
\item $\sA_{1}$ receives $b \in \{0, 1\}^{n}$, $\sA_{2}$ receives $d \in \{0, 1\}$ (both chosen uniformly at random).
\item $\sA_{1}$ outputs $a_{1} \in \{0, 1\}^{n}$, $\sA_{2}$ outputs $a_{2} \in \{0, 1\}^{n}$ and they win iff $a_{1} \oplus a_{2} = d \cdot b$.
\end{enumerate}
This game has been considered in Ref.~\cite{sikora14} under the name $\textnormal{CHSH}_{n}$ and it has been shown that
\begin{equation*}
p_{\textnormal{win}}(n) \leq \frac{1}{2} + \frac{1}{\sqrt{2^{n + 1}}},
\end{equation*}
which is sufficient for our purposes as it implies that
\begin{equation*}
p_{0} + p_{1} \leq 1 + \sqrt{2} \cdot 2^{- n / 2}
\end{equation*}
for all strategies of dishonest Alice. Therefore, the protocol is $\varepsilon$-binding with $\varepsilon = 2^{(1 - n) / 2}$ decaying exponentially in $n$ (but note that the decay rate is half of the decay rate against classical adversaries).

The two-round protocol is mapped onto a non-local game precisely because of the assumption of no communication. More specifically, we require that $\sA_{1}$ outputs the answer outside of the future of $\sA_{2}$ receiving the input and vice versa.
\paragraph*{A new multi-round protocol}
To extend the commitment time we propose a multi-round protocol and prove its security against classical adversaries. In principle, the commitment time can be made arbitrarily long. However, security depends on the number of rounds of the protocol, which is proportional to the length of the commitment. Therefore, the longer the commitment, the more resources (randomness and communication bandwidth) are required to achieve a given level of security.

Suppose that Alice and Bob want to execute the protocol with $m + 1$ rounds and we use $k$ as a label for the round under consideration. Then $\sA_{1}$ and $\sA_{2}$ must share $m$ secret strings denoted by $\{a_{k}\}_{k = 1}^{m}$. Similarly, Bob's agents need one secret string for every round denoted by $\{b_{k}\}_{k = 1}^{m}$ but, again, these can be generated locally during the protocol. All the rounds before the open phase ($1 \leq k \leq m$) have the same communication pattern: first $\sB_{i}$ sends an $n$-bit string to $\sA_{i}$ and then she replies with another $n$-bit string. In the last round $\sA_{i}$ sends $\sB_{i}$ a bit (her commitment) and an $n$-bit string (proof of her commitment). We will denote the $n$-bit string announced by Bob (Alice) in the $k\th$ round by $x_{k}$ ($y_{k}$) regardless of whether he/she is honest or not. The protocol is:
\begin{enumerate}
\item (commit, $k = 1$) $\sB_{1}$ sends $x_{1} = b_{1}$ to $\sA_{1}$. $\sA_{1}$ returns $y_{1} = d \cdot x_{1} \oplus a_{1}$.
\item (sustain, $2 \leq k \leq m$) $\sB_{i}$ sends $x_{k} = b_{k}$ to $\sA_{i}$. $\sA_{i}$ returns $y_{k} = (x_{k} * a_{k - 1}) \oplus a_{k}$.
\item (open, $k = m + 1$) $\sA_{i}$ sends $d$ and $y_{m + 1} = a_{m}$ to $\sB_{i}$.
\end{enumerate}
To check whether the commitment should be accepted $\sB_{1}$ and $\sB_{2}$ communicate and verify the following relation:
\begin{equation*}
\begin{aligned}
\label{eq:acceptance-condition}
y_{m + 1} \; = \; y_{m} \; \oplus \; b_{m} * y_{m - 1} \; \oplus \; b_{m} * b_{m - 1} * y_{m - 2} \; \oplus \; \ldots\\
\ldots \; \oplus \; b_{m} * b_{m - 1} * \ldots * b_{2} * y_{1} \; \oplus \; d \cdot b_{m} * b_{m - 1} * \ldots * b_{1}.
\end{aligned}
\end{equation*}

Security for honest Alice is a direct consequence of the fact that every message she announces is masked by a fresh secret $n$-bit string, which implies that the transcripts corresponding to $d = 0$ and $d = 1$ are statistically indistinguishable (see Proposition C.1 in the SM).

Proving security for honest Bob is a more challenging task, because we require security immediately after round $k = 1$. We first state the main result and then outline the idea behind the proof (for details refer to Sections B.2 and C.2 in the SM).
The multi-round protocol with $m + 1$ rounds is $\varepsilon$-binding for $\varepsilon = c_{m}$ defined as
\begin{equation} 
c_{m} =
\begin{cases}
2^{-n} &\nbox{for} m = 1,\\
\frac{1}{2^{n + 1}} + \sqrt{c_{m - 1}} &\nbox{for} m \geq 2.
\end{cases} \label{eq:boundmulti}
\end{equation} 
The security argument is conceptually simple: in the classical scenario the~\textit{sequential} cheating game in the multi-round protocol is \emph{equivalent} to a game in which multiple players act~\textit{in parallel} which allows us to disregard the causal structure of the protocol. We show that cheating in a protocol with $m + 1$ rounds is at least as difficult as winning the following $m$-player game. Let $X_{1}, X_{2}, \ldots, X_{m}$ be independent random variables drawn uniformly from the set of $n$-bit strings $\{0, 1\}^{n}$ and the $k\th$ player receives all the variables except for $X_{k}$ and outputs an $n$-bit string. The game is won if the XOR of the outputs equals $X_{1} * X_{2} * \ldots * X_{m}$. The bounds we obtain decay exponentially in $n/2^{m}$. This means that they become significantly weaker as the number of players increases, which ultimately limits the maximum number of rounds that can be implemented in practice. 
The tightness of these bounds is an interesting open problem and it is briefly discussed in Appendix B. Note that no explicit cheating strategy is known, whose winning probability would approach our security bounds.

\paragraph*{Implementation}
We implemented the two-round and the multi-round protocols described above.  
Each party has agents at two distinct locations: one at the Group of Applied Physics of the University of Geneva and one at the Institute of Applied Physics of the University of Berne. The straight-line distance between the two locations is $s=131$~km, corresponding to a time separation of 437~\textmu s. The hardware installed in Geneva is conceptually represented in FIG.~\ref{fig:setup}(a) and identical to the one in Berne. Each of the classical agents is a standalone computer equipped with a field-programmable gate array (FPGA) programmed to execute the necessary steps of the protocol. Each FPGA is synchronized to the Coordinated Universal Time (UTC) via a Global Positioning System clock (GPS clock), which consists of a GPS receiver and a Oven-Controlled Quartz-Crystal Oscillator (OCXO) generating a 10~MHz sinusoidal waveform. Through its GPS connection, the receiver outputs one electronic pulse per second (PPS), which is used to discipline the OCXO. The receiver is locked to the GPS signal with a time accuracy better than 150~ns. The 10~MHz signal generated by the OCXO is fed into the FPGA board and it is used to generate a 125~MHz signal using a phase-locked loop. This 125~MHz signal then serves as the time basis for the computations performed on the FPGA. The FPGA also receives the PPS signal to monitor the synchronization with the GPS clock. In particular, the number of cycles between two successive PPS signals is confirmed to be $125\times10^6$ plus or minus one, where each cycle corresponds to 8~ns. Therefore, the FPGA tolerates fluctuations up to 24~ns on the arrival time of the PPS synchronization signal. The GPS clock also provides the FPGA with a universal time stamp of every PPS signal, allowing Alice and Bob to locate their actions in time. 

\begin{figure}[!t]
\includegraphics[width= 6 cm]{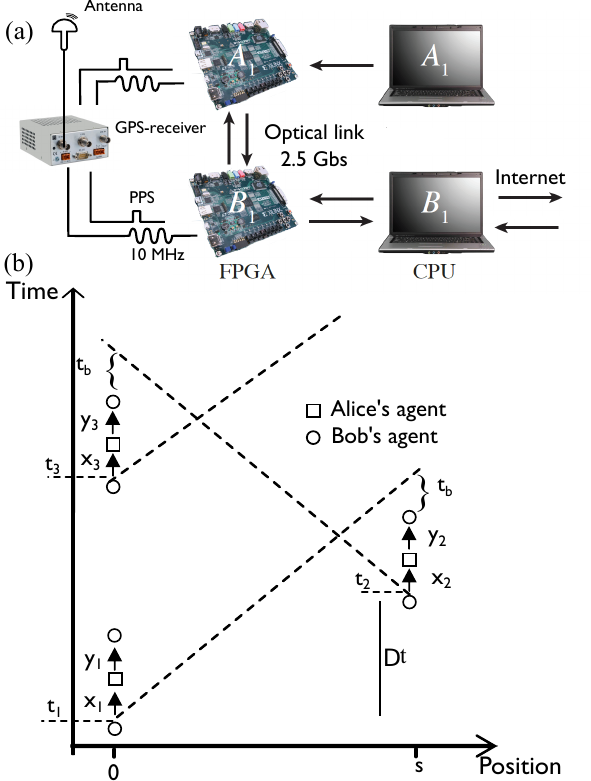}
\caption{(a) Experimental setup. (b) Space-time diagram of the experimental setup.}
\label{fig:setup}
\end{figure} 

Before either the two-round or the multi-round protocol starts, $\sA_{1}$ and $\sA_{2}$ (and similarly $\sB_{1}$ and $\sB_{2}$) share an appropriate number of random $n$-bit strings. At time $t_1$, which was agreed upon by both parties, $\sB_{1}$ sends the random string $x_1$ through the optical link. For a string of 512 bits communicated through the 2.5 Gbps optical link, this requires 205 ns. $\sA_{1}$'s FPGA then computes the string $y_1$ and sends it to $\sB_{1}$; see FIG.~\ref{fig:setup}(b). The relativistic constraint requires space-like separation between every two consecutive rounds, which means that the entire second round must be outside of the future light cone of the first bit of $x_1$ leaving the FPGA of $\sB_{1}$. The commitment begins when the last bit of $y_1$ is recorded by the FPGA of $\sB_{1}$. With $n=512$~bits, the security parameter of the two-round protocol is $\varepsilon \approx 10^{-77}$.

In the two-round protocol, $\sA_{2}$ unveils the commitment in the second round, at time $t_2 = t_1 + \Delta t$. She does so by sending the string $a_1$ to $\sB_{2}$, along with the committed bit~$d$. $\sB_{2}$ checks that the last bit of $a_1$ is received outside the future light cone of the beginning of the protocol. If this is the case, $\sB_{2}$ communicates $a_1$ and $d$ to $\sB_{1}$ through an authenticated channel. Finally, $\sB_{1}$ verifies that $y_1\oplus a_1 = d\cdot x_1$ and accepts the commitment. If the relativistic constraint is not respected, or if $\sB_{1}$'s verification fails, the protocol aborts.

In the multi-round protocol, $\sA_{1}$ and $\sA_{2}$ successively sustain the commitment until the last round. All rounds (except the first and last rounds) proceed as follows. Let us consider the $k^{\text{th}}$ round, with $k$ even (odd rounds are similar).  Between rounds $k$ and $k-2$, the string $x_k$ is loaded in the memory of $\sB_{2}$'s FPGA, and strings $a_{k-1}$ and $a_k$ are loaded in $\sA_{2}$'s FPGA. At time $t_k = t_1 + (k-1)\Delta t$,  $\sB_{2}$ communicates $x_k$ through the optical link. Then $\sA_{2}$ sustains the commitment by computing $y_k$ with the FPGA and sending it to $\sB_{2}$. The time between the communication of $x_k$ and the reception of $y_k$ is 6.1~\textmu s. $\sB_{2}$ checks that the reception of $y_k$ is outside the future light cone of the beginning of the communication between $\sB_{1}$ to $\sA_{1}$ that happened in round $k-1$. We used $\Delta t =  400 $~\textmu s (see Fig.~\ref{fig:setup}), which is 37~\textmu s shorter than the light-speed separation between the Berne and Geneva locations. Considering the 6.1~\textmu s, the absolute inaccuracies of the GPS clock ($\le 150$~ns), and the tolerance in the fluctuations of the synchronization signals ($\le 24$~ns) the round is completed $\approx 30.7$~\textmu s before the relativistic constraint expires. In the final $(m+1)^{\text{th}}$ round, $\sA_{1}$ (or $\sA_{2}$) opens the commitment at time $t_{m+1}$ by sending the string $a_{m+1}$ along with the committed bit~$d$. To verify the commitment, $\sB_{2}$ sends to $\sB_{1}$ all the strings communicated by $\sA_{2}$ through an authenticated channel. $\sB_{1}$ then checks if the commitment should be accepted as outlined above.
Authentication is based on an information-theoretic secure message-authenticator code which consists of a combination of polynomial hashing, and a strongly-universal family of hash functions~\cite{Carter1979}.

In the multi-round scheme, we aimed to maximize the number of rounds with a reasonable value for the security parameter $\varepsilon$. The limit of $n=512$~bits and $m+1=6$~rounds was ultimately set by the performance that we could achieve with the FPGA at our disposal. This yields a security parameter of $\varepsilon \approx 2.3\times 10^{-10}$. Using these parameters, we realized a commitment of 2~ms duration, which extends beyond the 437~\textmu s limit of the two-round protocol. Because synchronizing rounds over longer durations is a simple task for our hardware, it is straightforward to achieve significantly longer commitment times using more distant agents. For example, 150~ms could be easily achieved using Geneva and Singapore as the locations (these locations were used in our previous demonstration of quantum-relativistic bit commitment~\cite{Lunghi2013a}), while 212~ms could be achieved using antipodal locations on the Earth. 

\paragraph*{Summary}
We have shown that classical relativistic protocols allow us to implement information-theoretically secure commitment schemes in a straightforward fashion. 

The commitment scheme we implemented belongs to the class of timed commitments, i.e.~commitments that expire after a certain period of time. Even though they cannot be used to implement primitives whose security is required to hold forever (e.g.~oblivious transfer), they are known to have other important applications, e.g.~contract signing, honesty-preserving auctions or secure voting~\cite{Boneh2000a, Broadbent2008a} (see also Appendix~A).

For the \texttt{sBGKW} scheme we obtain an explicit, quantitative security bound by making a connection to a non-local game analyzed previously. We also propose a multi-round scheme which is secure against classical adversaries. We note that the number of rounds that we implemented here could have been higher using better optimized hardware. However, the scaling of the security bound with the number of rounds~\eqref{eq:boundmulti} prohibits a much larger number of rounds. An important problem is therefore to find a multi-round protocol whose security exhibits better scaling with the number of rounds, or, ideally, no dependence at all. This would allow us to obtain longer (or maybe even arbitrarily long) commitments while only using simple, commercially available digital devices.

\paragraph*{Acknowledgments:} We thank Mohammad Bavarian, Gilles Brassard, Iordanis Kerenidis, Raghav Kulkarni, Troy Lee, Laura Man\v{c}inska, Miklos Santha and Sarvagya Upadhyay for useful discussions. JK especially thanks Igor Shparlinski for sharing his ideas about Proposition B.2 of the SM and subsequent discussions. We thank Andr\'{e} Stefanov and Daniel Weber for helping to install the setup in Berne. JK, MT and SW are funded by the Ministry of Education (MOE) and National Research Foundation Singapore, as well as MOE Tier 3 Grant "Random numbers from quantum processes" (MOE2012-T3-1-009). Financial support is provided by the Swiss NCCR QSIT.

\paragraph*{Author contributions:} TL is the first experimental author and JK is the first theoretical author.
\nocite{buhrman05}
\nocite{bavarian13}
\nocite{ford13}
\bibliographystyle{apsrev4-1}
\bibliography{qcrypt2014_v2}
\onecolumngrid
\appendix
\section{Preliminaries}
\subsection{How useful is a relativistic bit commitment protocol?} The commitment scheme we implement belongs to the class of timed commitments, i.e.~commitments that are only valid for a period of time but then ultimately \emph{expire}. Such commitments cannot be used in reductions implementing primitives whose security is required to hold forever (e.g. oblivious transfer or secure function evaluation) but they are known to have other applications. For example, Boneh and Naor~\citep{Boneh2000a} study commitments which after some fixed $\rm{\Delta}$t automatically open, i.e.~the committed value is revealed. They show that such commitments can be used for contract signing or honesty-preserving auctions. In our case after $\rm{\Delta}$t the commitment simply vanishes, i.e.~Bob should not accept any opening (because Alice could unveil either bit with unit probability) and the originally committed value (if there was one) remains secret. Therefore, it gives more power to the committer by giving her the freedom not to open the commitment and, hence, protect her privacy. Generally speaking, such temporary secrecy is sufficient if the goal is not to preserve secrecy forever but to force parties to act simultaneously (in the sense that their respective actions should not depend on each other) even if the communication model is sequential. Our commitment might be particularly useful for multi-party protocols which are robust against a certain fraction of dishonest parties (then we would simply call dishonest any party that refuses to open the commitment). A prime application of this type would be the task of secure voting as presented in Ref.~\cite{Broadbent2008a}.

\subsection{Notation}
Let $[n] = \{1, 2, \ldots, n\}$. Generally, we use uppercase letters for random variables and lowercase letters to denote particular values. For $j, k \in \cS$ we use $\sum_{j \neq k}$ as shorthand notation for $\sum_{j \in \cS} \sum_{k \in \cS \setminus \{j\}}$.
\subsection{The Cauchy-Schwarz inequality for probabilities}
Let $X$ be a random variable distributed uniformly over $[n]$ and let $\{E_{j}\}_{j \in [m]}$ be a family of events defined on $X$.
\begin{lem}
\label{lem:cauchy-schwarz}
Let $p$ be the average probability of the family of events
\begin{equation*}
p := \frac{1}{m} \sum_{j \in [m]} \Pr[ E_{j} ]
\end{equation*}
and $c$ be the cumulative size of the pairwise intersections
\begin{equation*}
c := \sum_{j \neq k} \Pr[ E_{j} \wedge E_{k} ].
\end{equation*}
Then the following inequality holds
\begin{equation*}
p \leq \frac{ 1 + \sqrt{1 + 4c} }{2m}.
\end{equation*}
\end{lem}
\begin{proof}
Each event can be represented by an $n$-dimensional, real vector whose entries are labelled by the possible values that $X$ can take. If a particular value of $X$ belongs to the event, we set the corresponding component to $1/\sqrt{n}$ and if it does not we set it to $0$
\begin{equation*}
[ s_{j} ]_{x} =
\begin{cases}
\frac{1}{\sqrt{n}} &\nbox[6]{if} x \in E_{j},\\
0 &\nbox[6]{otherwise.}
\end{cases}
\end{equation*}
Moreover, let $n$ be the normalised, uniform vector: $[n]_{x} = 1/\sqrt{n}$ for all $x \in [n]$. It is straightforward to check that with these definitions we have
\begin{equation*}
\Pr[ E_{j} ] = \ave{s_{j}, n} = \ave{s_{j}, s_{j}} \nbox[8]{and} \Pr[ E_{j} \wedge E_{k} ] = \ave{s_{j}, s_{k}}.
\end{equation*}
Clearly, we have $\ave{s_{j}, s_{k}} \geq 0$. Due to linearity of the inner product we have
\begin{equation*}
p = \frac{1}{m} \sum_{j \in [m]} \Pr[ E_{j} ] = \frac{1}{m} \sum_{j} \ave{s_{j}, n} = \frac{1}{m} \ave{ \sum_{j} s_{j}, n},
\end{equation*}
which can be upper bounded using the Cauchy-Schwarz inequality. Since $\ave{n, n} = 1$ we have
\begin{equation*}
\ave{ \sum_{j} s_{j}, n}^{2} \leq \sum_{j k} \ave{s_{j}, s_{k}} = \sum_{j} \ave{s_{j}, s_{j}} + \sum_{j \neq k} \ave{s_{j}, s_{k}} = m p + c,
\end{equation*}
which gives the following quadratic constraint
\begin{equation*}
p^{2} \leq \frac{p}{m} + \frac{c}{m^{2}}.
\end{equation*}
Solving for $p$ gives the desired bound.
\end{proof}
\section{Finite-field multiplication in the ``Number on the Forehead'' model}
We introduce a family of multiplayer games which are a natural generalisation of the two-player family introduced in Ref.~\cite{buhrman05} and generalised in Ref.~\cite{bavarian13}. Since these games rely on finite-field arithmetic we first state some basic properties of finite fields, then we define the game and show that finding the optimal winning probability corresponds to a natural algebraic problem concerning multivariate polynomials over finite fields. Finally, we prove upper bounds on the optimal winning probability and discuss their tightness.
\subsection{Finite-field arithmetic}
Let $\amsbb{F}_{q}$ denote the finite field of order $q = p^{k}$ ($p$ is a prime and $k$ is an integer) and let $0$ denote the zero element of $\amsbb{F}_{q}$. Operations over finite-field satisfy the following properties
\begin{enumerate}
\item Multiplication by zero gives zero\\
$x \cdot 0 = 0 \quad \forall x \in \amsbb{F}_{q}$
\item Multiplication is distributive over addition\\
$x (y + z) = (x y) + (x z) \quad \forall x, y, z \in \amsbb{F}_{q}$
\end{enumerate}
\subsection{Definition of the game}
\label{sec:multiplayer-game}
\begin{figure}
\begin{tikzpicture}[scale=1, line width=0.5]
\draw (-0.4, -0.4) rectangle (0.4, 0.4);
\node at (0, 0) {$\cP_{1}$};
\draw[-] (1.5, 1.5) -- (1.5, -1.5);
\draw (2.6, -0.4) rectangle (3.4, 0.4);
\node at (3, 0) {$\cP_{2}$};
\draw[-] (4.5, 1.5) -- (4.5, -1.5);
\hordots{5.5}{0}{0.05}{black}
\draw[-] (6.5, 1.5) -- (6.5, -1.5);
\draw (7.6, -0.4) rectangle (8.4, 0.4);
\node at (8, 0) {$\cP_{m}$};
\node at (0, 1.4) {$X_{2}, X_{3}, \ldots, X_{m}$};
\draw [->] (0, 1.1) to (0, 0.6);
\node at (3, 1.4) {$X_{1}, X_{3}, \ldots, X_{m}$};
\draw [->] (3, 1.1) to (3, 0.6);
\node at (8, 1.4) {$X_{1}, X_{2}, \ldots, X_{m - 1}$};
\draw [->] (8, 1.1) to (8, 0.6);
\node at (0, -1.4) {$Y_{1}$};
\draw [->] (0, -0.6) to (0, -1.1);
\node at (3, -1.4) {$Y_{2}$};
\draw [->] (3, -0.6) to (3, -1.1);
\node at (8, -1.4) {$Y_{m}$};
\draw [->] (8, -0.6) to (8, -1.1);
\end{tikzpicture}
\caption{In the ``Number on the Forehead'' model there are $m$ inputs $X_{1}, X_{2}, \ldots, X_{m}$ and $\cP_{k}$ (the $k\th$ player out of $m$) receives all the inputs except for $X_{k}$. We denote the output of $\cP_{k}$ by $Y_{k}$. Vertical lines remind us that no communication between the players is allowed.}
\label{fig:number-on-the-forehead}
\end{figure}
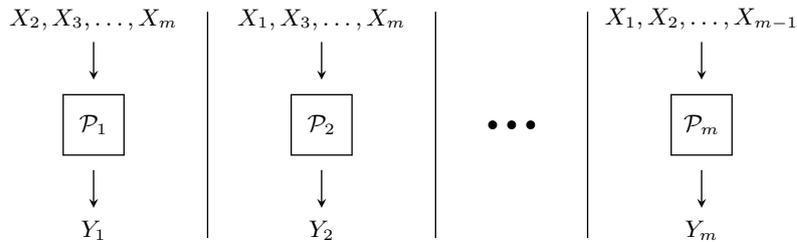
Consider a one-round game with $m$-players denoted by $\cP_{1}, \cP_{2}, \ldots, \cP_{m}$. With every player we associate an \emph{input} and an \emph{output}, e.g.~for $\cP_{k}$ these are denoted by $X_{k}$ and $Y_{k}$, respectively. Let each of $X_{1}, X_{2}, \ldots, X_{m}$ be drawn independently, uniformly at random from $\amsbb{F}_{q}$. In the ``Number on the Forehead'' model $\cP_{k}$ receives \emph{all the inputs except for the $k\th$ one} (denoted by $X_{[m] \setminus \{k\}}$) as shown in FIG.~\ref{fig:number-on-the-forehead}. Each player is required to output an element of $\amsbb{F}_{q}$ (denoted by $Y_{k}$) and the game is won if
\begin{equation}
\label{eq:winning-condition}
\prod_{k = 1}^{m} X_{k} = \sum_{k = 1}^{m} Y_{k}.
\end{equation}
In the classical setting the optimal winning probability can be achieved when each player adopts a deterministic strategy, i.e.~a function $f : \amsbb{F}_{q}^{(m - 1)} \to \amsbb{F}_{q}$. If $\cP_{k}$ employs a strategy represented by $f_{k}$, i.e.~he outputs $Y_{k} = f_{k} ( X_{ [m] \setminus \{k\} } )$, then the winning probability equals
\begin{equation*}
\omega_{m}(f_{1}, f_{2}, \ldots, f_{m}) := \Pr \big[ \prod_{k = 1}^{m} X_{k} = \sum_{k = 1}^{m} f_{k}( X_{ [m] \setminus \{k\} } ) \big]
\end{equation*}
and we define $\omega_{m}$ to be the optimal winning probability
\begin{equation}
\label{eq:omega-definition}
\omega_{m} := \max_{ f_{1}, f_{2}, \ldots, f_{m} } \omega_{m}(f_{1}, f_{2}, \ldots, f_{m}),
\end{equation}
where the maximisation is taken over all functions from $\amsbb{F}_{q}^{(m - 1)}$ to $\amsbb{F}_{q}$.
\subsection{Characterisation through multivariate polynomials over a finite field}
First, note that since the probability distribution of inputs is uniform then the winning probability is proportional to the number of inputs $(x_{1}, x_{2}, \ldots, x_{m})$ on which the condition~\eqref{eq:winning-condition} is satisfied
\begin{equation*}
\prod_{k = 1}^{m} x_{k} = \sum_{k = 1}^{m} f_{k}( x_{ [m] \setminus \{k\} } ).
\end{equation*}
Alternatively, the winning probability can be deduced by counting the number of zeroes of the following function
\begin{equation*}
P(x_{1}, x_{2}, \ldots, x_{m}) = \prod_{k = 1}^{m} x_{k} - \sum_{k = 1}^{m} f_{k}( x_{ [m] \setminus \{k\} } ).
\end{equation*}
By the Lagrange interpolation method every function from $\amsbb{F}_{q}^{n} \to \amsbb{F}_{q}$ (for arbitrary $n \in \amsbb{N}$) can be written as a polynomial. Therefore, the question concerns the number of zeroes of the polynomial $P$. Different strategies employed by the players give rise to different polynomials and we need to characterise what polynomials are ``reachable'' in this model. The output of $\cP_{k}$ is an arbitrary polynomial of $x_{ [m] \setminus \{k\}}$, hence, it only contains terms that depend on \emph{at most $m - 1$ variables}. This means that the part of $P$ that depends on \emph{all $m$ variables} comes solely from the first term and equals $\prod_{k = 1}^{m} x_{k}$. Therefore, finding the optimal winning probability of the game is equivalent to finding the polynomial with the largest number of zeroes, whose only term that depends on all $m$ variables equals $\prod_{k = 1}^{m} x_{k}$. This reduces the problem of finding the optimal strategy to a purely algebraic problem about properties of polynomials over finite fields.
\subsection{A recursive upper bound on the optimal winning probability}
\label{sec:recursive-upper-bound}
Here, we show how to find explicit upper bounds on $\omega_{m}$ through an induction argument. First, note that for $m = 1$ there is only one term on the right-hand side of Eq.~\eqref{eq:winning-condition} and since this term takes no arguments it is actually a constant. Since $X_{1}$ is uniform we have
\begin{equation*}
\omega_{1} := \max_{ c \in \amsbb{F}_{q} } \Pr [ X_{1} = c ] = \frac{1}{q}.
\end{equation*}

Now, we show how to prove an upper bound on $\omega_{m}$ in terms of $\omega_{m - 1}$. For a fixed strategy $\{f_{1}, f_{2}, \ldots, f_{m}\}$ the winning probability can be written as
\begin{align*}
\omega_{m}(f_{1}, f_{2}, \ldots, f_{m}) &= \Pr \big[ X_{1} X_{2} \ldots X_{m} = \sum_{k = 1}^{m} f_{k}( X_{ [m] \setminus \{k\} } ) \big]\\
&= \sum_{y \in \amsbb{F}_{q}} \Pr[X_{m} = y] \cdot \Pr \big[ X_{1} X_{2} \ldots X_{m} = \sum_{k = 1}^{m} f_{k}( X_{ [m] \setminus \{k\} } ) | X_{m} = y \big]\\
&= q^{-1} \sum_{y} \Pr \big[ X_{1} X_{2} \ldots X_{m} = \sum_{k = 1}^{m} f_{k}( X_{ [m] \setminus \{k\} } ) | X_{m} = y \big].
\end{align*}
Conditioning on a particular value of $X_{m}$ leads to events that only depend on $X_{1}, X_{2}, \ldots, X_{m - 1}$. In particular, we can define for $X_{m} = y$ the event $F_{y}$
\begin{equation*}
F_{y} \iff X_{1} X_{2} \ldots X_{m - 1} y = \sum_{k = 1}^{m - 1} f_{k}( X_{ [m - 1] \setminus \{k\} }, y ) + f_{m}(X_{[m - 1]}),
\end{equation*}
which satisfies
\begin{equation}
\label{eq:fy-definition}
\Pr[ F_{y} ] = \Pr \big[ X_{1} X_{2} \ldots X_{m} = \sum_{k = 1}^{m} f_{k}( X_{ [m] \setminus \{k\} } ) | X_{m} = y \big].
\end{equation}
We can use Lemma~\ref{lem:cauchy-schwarz} to find a bound on $\omega_{m}(f_{1}, f_{2}, \ldots, f_{m}) = q^{-1} \sum_{y} \Pr[ F_{y} ]$ as long as we are given bounds on $\Pr[ F_{y} \wedge F_{z}]$ for $y \neq z$.
\begin{prop}
\label{prop:intersections}
For $y \neq z$ we have $\Pr[ F_{y} \wedge F_{z}] \leq \omega_{m - 1}$.
\end{prop}
\begin{proof}
Eq.~\eqref{eq:fy-definition} defines $F_{y}$ through a certain equation in the finite field. If the equations corresponding to $F_{y}$ and $F_{z}$ are satisfied simultaneously then clearly any linear combination of these equations is also satisfied. More specifically, we define a new event
\begin{equation}
\label{eq:gyz-definition}
G_{yz} \iff X_{1} X_{2} \ldots X_{m - 1} (y - z) = \sum_{k = 1}^{m - 1} f_{k}( X_{ [m - 1] \setminus \{k\} }, y ) - f_{k}( X_{ [m - 1] \setminus \{k\} }, z )
\end{equation}
and since $F_{y} \wedge F_{z} \implies G_{yz}$ we are guaranteed that $\Pr[F_{y} \wedge F_{z}] \leq \Pr[G_{yz}]$.

To find an upper bound on $\Pr[ G_{yz} ]$ we give the players more power by allowing a more general expression on the right-hand side. In Eq.~\eqref{eq:gyz-definition} the $k\th$ term is a function of $X_{ [m - 1] \setminus \{k\} }, y$ and $z$, so let us replace it by an arbitrary function of these variables
\begin{equation*}
f_{k}( X_{ [m - 1] \setminus \{k\} }, y ) - f_{k}( X_{ [m - 1] \setminus \{k\} }, z ) \quad \to \quad g_{k}( X_{ [m - 1] \setminus \{k\} }, y, z ).
\end{equation*}
Under this relaxation, we arrive at the following equality
\begin{equation*}
X_{1} X_{2} \ldots X_{m - 1} (y - z) = \sum_{k = 1}^{m - 1} g_{k}( X_{ [m - 1] \setminus \{k\} }, y, z ).
\end{equation*}
Clearly, $(y - z)$ is a constant (non-zero) multiplicative factor known to each player. Dividing the equation through by $(y - z)$ leads to the same game as considered before but the number of players has decreased by one: there are only $m - 1$ players now. Therefore,
\begin{equation*}
\Pr[ F_{y} \wedge F_{z} ] \leq \Pr[ G_{yz} ] \leq \omega_{m - 1}.
\end{equation*}
\end{proof}
Now, we can state and prove our main technical result.
\begin{prop}
The optimal winning probability of the game defined in~\eqref{eq:omega-definition} satisfies the following recursive relation
\begin{equation}
\label{eq:recursive-bound}
\omega_{m} \leq \frac{1 + \sqrt{1 + 4 q ( q - 1 ) \omega_{m - 1} } }{ 2 q }.
\end{equation}
\end{prop}
\begin{proof}
The statement follows directly from combining Lemma~\ref{lem:cauchy-schwarz} with Proposition~\ref{prop:intersections}.
\end{proof}
Since we know that $\omega_{1} = q^{-1}$, we can obtain a bound on $\omega_{m}$ by recursive evaluation of Eq.~\eqref{eq:recursive-bound}. More precisely, we get $\omega_{m} \leq c_{m}$ for
\begin{equation}
\label{eq:cm-definition}
c_{m} =
\begin{cases}
q^{-1} &\nbox{for} m = 1,\\
\frac{1 + \sqrt{1 + 4 q ( q - 1 ) c_{m - 1} } }{ 2 q } &\nbox{for} m \geq 2.
\end{cases}
\end{equation}
Note that this bound is always non trivial, i.e.~$c_{m} < 1$ for all values of $q$ and $m$. To obtain a slightly weaker but simpler form presented as Eq.~(1) in main text we note that $1 - 4 q c_{m - 1} \leq 0$ and set $q = 2^{n}$.
\subsection{Discussion}
Having proved an explicit upper bound on $\omega_{m}$ we would like to investigate its tightness. It can be shown that in the regime interesting from the cryptographic point of view ($q \gg 1$) the leading behaviour of $c_{m}$ is
\begin{equation}
\label{eq:cm-asymptotic}
c_{m} \propto q^{-2^{-m}}.
\end{equation}
In other words, $c_{m}$ decays exponentially in $q$ but the value of the exponent depends on the number of players: every time we add a player we lose half of the decay exponent. This might seem unexpectedly weak but it has recently been shown (Theorem 6.5 in Ref.~\cite{bavarian13}) that
\begin{equation*}
\omega_{2} = \Omega(q^{- \frac{2}{3}}).
\end{equation*}
In fact, for $q = p^{k}$ where $k$ is even it can be improved (Theorem 1.3 in Ref.~\cite{bavarian13}) to give
\begin{equation*}
\omega_{2} = \Omega(q^{- \frac{1}{2}}).
\end{equation*}
This shows that for $m = 2$ the asymptotic decay of $q^{-1/2}$ is the best we can hope for (at least for an upper bound that holds for both odd and even values of $q$). Moreover, as far as we know, the best explicit upper bound on $\omega_{2}$ is the quantum upper bound (Theorem 1.2 in Ref.~\cite{bavarian13})
\begin{equation*}
c_{2}^{\textnormal{qm}} = \frac{1}{q} + \frac{q - 1}{q} \frac{1}{ \sqrt{q} }.
\end{equation*}
On the other hand, evaluating $c_{2}$ according to Eq.~\eqref{eq:cm-definition} gives
\begin{equation*}
c_{2} = \frac{ \sqrt{1 - \frac{3}{4q} }}{ \sqrt{q} } + \frac{1}{2q}.
\end{equation*}
By noting that
\begin{gather*}
\sqrt{1 - \frac{3}{4q} } \leq 1 - \frac{3}{8q} \nbox{for} q \geq \frac{3}{4} \nbox{and}\\
\frac{1}{\sqrt{q}} - \frac{3}{8q \sqrt{q}} + \frac{1}{2q} \leq \frac{1}{\sqrt{q}} + \frac{1}{q} - \frac{1}{q \sqrt{q}} \nbox{for} q \geq \frac{25}{16}
\end{gather*}
we conclude that our bound is strictly tighter, $c_{2} < c_{2}^{\textnormal{qm}}$, for all $q \geq 2$. (Note that since $c_{2}$ only applies to classical strategies, this comparison has no implications on the tightness of $c_{2}^{\textnormal{qm}}$ in the quantum setting.)

To close the discussion, let us mention that the ``Number on the Forehead'' model has been extensively studied in the communication complexity setting. In fact, for certain Boolean functions related to finite-field multiplication (in finite fields of characteristic $2$, i.e.~$q = 2^{n}$) a lower bound of the form $\Omega(n / 2^{m})$ was recently shown \cite{ford13}. As it appears strikingly similar to~\eqref{eq:cm-asymptotic} it would be interesting to investigate whether the two scenarios can be related in a rigorous way.
\section{Relativistic bit commitment protocols}
All the $n$-bit strings that appear in the protocols below should be interpreted as elements of the finite field $\amsbb{F}_{2^{n}}$. The addition and multiplication are denoted by $\oplus$ and $*$, respectively. Note that the addition is exactly the bitwise XOR but the multiplication \emph{does not} correspond to taking bitwise AND. Moreover, if $d$ is a bit and $b$ is an $n$-bit string we define
\begin{equation*}
d \cdot b =
\begin{cases}
0 &\nbox[6]{if} d = 0,\\
b &\nbox[6]{if} d = 1.
\end{cases}
\end{equation*}
\subsection{Causal structure of the protocol}
\label{sec:causal-structure}
The relativistic protocols we consider require Alice (who makes the commitment) and Bob (who receives the commitment) to delegate agents to exchange information at two distant locations labelled by $1$ and $2$. We refer to Alice's (Bob's) agent at the $i\th$ location by $\sA_{i}$ ($\sB_{i}$). Odd (even) rounds take place at Location 1 (2) and the timing is chosen such that every pair of consecutive rounds is space-like separated (see FIG.~\ref{fig:causal-structure}), which means that Alice's message in the $(k + 1)\th$ round must not depend on Bob's message in the $k\th$ round and vice versa. These are the \emph{causal constraints} of the protocol and any actions that do not violate them are allowed. Note that these constraints are \emph{strictly more restrictive} than any form of no-signalling between the two locations. We will see in Section C.4 how these constraints restrict the power of the dishonest party.
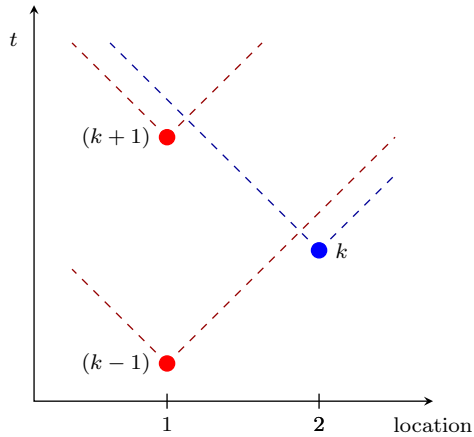
\begin{figure}
\begin{tikzpicture}[scale=1, line width=0.5, font=\footnotesize]
\draw [<->] (1.25, 5.75) -- (1.25, 0.5) -- (6.5, 0.5);
\node[left] at (1.15, 5.3) {$t$};
\draw (3, 0.4) to (3, 0.6);
\node[below] at (3, 0.4) {$1$};
\draw (5, 0.4) to (5, 0.6);
\node[below] at (5, 0.4) {$2$};
\node[below] at (5, 0.4) {$2$};
\node[below] at (6.5, 0.4) {location};
\draw[dashed, color=darkred] (3, 1) to (1.75, 2.25);
\draw[dashed, color=darkred] (3, 1) to (6, 4);
\draw[fill, color=red] (3, 1) circle [radius=0.1];
\node[left] at (2.9, 1) {$(k-1)$};
\draw[dashed, color=darkblue] (5, 2.5) to (2.25, 5.25);
\draw[dashed, color=darkblue] (5, 2.5) to (6, 3.5);
\draw[fill, color=blue] (5, 2.5) circle [radius=0.1];
\node[right] at (5.1, 2.5) {$k$};
\draw[dashed, color=darkred] (3, 4) to (1.75, 5.25);
\draw[dashed, color=darkred] (3, 4) to (4.25, 5.25);
\draw[fill, color=red] (3, 4) circle [radius=0.1];
\node[left] at (2.9, 4) {$(k+1)$};
\end{tikzpicture}
\caption{The protocol requires each pair of consecutive rounds to be space-like separated. Here, we show a valid space-time arrangement for rounds $(k - 1)$, $k$ and $(k + 1)$. Dots represent rounds and dashed lines represent future light cones. Clearly, both pairs ($k - 1,k$) and $(k, k + 1)$ are space-like separated.}
\label{fig:causal-structure}
\end{figure}
\subsection{Security model and definitions}
Here, we describe the class of cheating strategies allowed for the dishonest parties and how to quantify security for the honest ones.
\subsubsection{Honest Alice and dishonest Bob}
\label{sec:honest-alice}
If Alice is honest, she will make an honest commitment to $d$ and Bob should remain completely ignorant about $d$ until the open phase \emph{regardless of his behaviour during the protocol} (for an explanation of the phase structure of relativistic bit commitment schemes see Ref.~\cite{Kaniewski2013a}). More specifically, we define the knowledge of Bob as the knowledge of $\sB_{1}$ and $\sB_{2}$ pooled together. (To see why it is not sufficient to require that \emph{each} agent remains ignorant note that if Alice aborts the protocol immediately before the open phase we want Bob to remain ignorant about $d$ for an indefinite period of time. Clearly, in this setting there is enough time for $\sB_{1}$ and $\sB_{2}$ to combine their knowledge.) Dishonest Bob is limited only by the causal constraints (explained in Section C.1): the messages announced by his classical agents may be arbitrary functions of all messages exchanged in the past (including any randomness generated before the protocol begins available to both $\sB_{1}$ and $\sB_{2}$). Quantum agents are, in addition, allowed to preshare a quantum state (of arbitrary dimension) and then throughout the protocol perform arbitrary measurements on it.

The security definition for honest Alice is based on the \emph{transcript}, which contains complete information about all the messages exchanged in the protocol at both locations and, hence, represents the combined knowledge of $\sB_{1}$ and $\sB_{2}$. Since in our case all the messages are classical the transcript is just a classical random variable denoted by $T$. Perfect security for Alice means that Bob should not be able to extract any information about her commitment. More precisely, we require that the distributions of transcripts are identical (statistically indistinguishable). We say that a protocol is \emph{perfectly hiding} if
\begin{equation*}
\Pr[T = t | d = 0] = \Pr[T = t | d = 1] \quad \forall t
\end{equation*}
for all strategies of dishonest Bob. Note that $d$ is not a random variable, it is Alice's input to the protocol. Therefore, $\Pr[T = t | d = 0]$ should be understood as the probability of seeing the transcript $t$ \emph{given} that Alice has decided to commit to $0$.
\subsubsection{Honest Bob and dishonest Alice}
\label{sec:honest-bob}
If Bob is honest, no strategy of dishonest Alice should allow her to successfully unveil both bits with high probability. Dishonest Alice is, again, limited only by the causal constraints: the messages she announces may be arbitrary functions of all messages exchanged in the past (including any randomness generated before the protocol begins available to both $\sA_{1}$ and $\sA_{2}$). Since we want our protocol to force Alice to become committed in the commit phase we must show that even if Alice decides on the value of the commitment \emph{immediately} after the commit phase she will still fail. Therefore, we consider a model in which her behaviour in the commit phase must be independent of the bit she will attempt to unveil later, denoted by $d$, but the messages exchanged later might depend on it. Given a particular strategy adopted by Alice in the commit phase we define $p_{d}$ to be the optimal probability of successfully unveiling $d$. The protocol is $\varepsilon$-\emph{binding} if
\begin{equation*}
p_{0} + p_{1} \leq 1 + \varepsilon
\end{equation*}
for all strategies of dishonest Alice in the commit phase. Note that this is a weak, non-composable definition of security and that in some other scenarios stronger security definitions can be used~\cite{Kaniewski2013a}.
\subsection{Security of \texttt{sBGKW} scheme against quantum adversaries}
\label{sec:simards-scheme}
The protocol requires $\sA_{1}$ and $\sA_{2}$ to share a secret $n$-bit string, $a \in \{0, 1\}^{n}$, chosen uniformly at random, which is consumed in the protocol.
\begin{enumerate}
\item (commit) $\sB_{1}$ sends $\sA_{1}$ an $n$-bit string, $b \in \{0, 1\}^{n}$, chosen uniformly at random. $\sA_{1}$ returns $d \cdot b \oplus a$ to $\sB_{1}$.
\item (open) $\sA_{1}$ unveils to $\sB_{1}$ the committed bit $d$ while $\sA_{2}$ unveils to $\sB_{2}$ the secret string $a$.
\item (verify) Bob collects data from $\sB_{1}$ and $\sB_{2}$ and accepts the commitment iff the string returned by $\sA_{1}$ in the commit phase equals $d \cdot b \oplus a$.
\end{enumerate}

Security for honest Alice follows from the fact that the knowledge of Bob (more precisely, the knowledge of his two agents pooled together) before the open phase is restricted to one $n$-bit string: the string that $\sB_{1}$ receives in the commit phase, which in the honest scenario equals $d \cdot b \oplus a$. It follows that as long as Alice is honest and chooses $a$ uniformly this string is distributed uniformly regardless of the value of her commitment.

Security for honest Bob against dishonest Alice who is restricted to classical cheating strategies is fairly intuitive: in order for $\sA_{2}$ to be able to unveil both commitments she would need to know both $a$ and $a \oplus b$, which implies that she would know $b$. However, since $b$ is chosen uniformly at random by Bob this must be difficult. This argument can be made rigorous~\cite{crepeau11} to show that the protocol is $\varepsilon$-binding for $\varepsilon = 2^{-n}$ (and this is actually tight: the trivial strategy of always outputting the string of all zeroes gives $p_{0} = 1$ and $p_{1} = 2^{-n}$). Unfortunately, this reasoning does not work against quantum adversaries since it could be the case that $\sA_{2}$ has two distinct measurements (one that reveals $a$ and another one that reveals $a \oplus b$) but since they are incompatible this does not imply anything about her ability to guess $b$.

To find an explicit bound on $p_{0} + p_{1}$ in the quantum setting we formulate cheating as a non-local game in which $\sA_{1}$ receives $b$, $\sA_{2}$ receives $d$ (the bit she is required to unveil, chosen uniformly at random) and the XOR of their outputs is supposed to equal $d \cdot b$. Winning such a game with probability $p_{\textnormal{win}}$ corresponds to a cheating strategy that satisfies $p_{0} + p_{1} = 2 p_{\textnormal{win}}$. More concisely, the rules of the non-local game are
\begin{enumerate}
\item $\sA_{1}$ receives $b \in \{0, 1\}^{n}$, $\sA_{2}$ receives $d \in \{0, 1\}$ (both chosen uniformly at random).
\item $\sA_{1}$ outputs $a_{1} \in \{0, 1\}^{n}$, $\sA_{2}$ outputs $a_{2} \in \{0, 1\}^{n}$ and they win iff $a_{1} \oplus a_{2} = d \cdot b$.
\end{enumerate}
This is exactly the game considered in Ref.~\cite{sikora14} under the name $\textnormal{CHSH}_{n}$
. They show that
\begin{equation*}
p_{\textnormal{win}}(n) \leq \frac{1}{2} + \frac{1}{\sqrt{2^{n + 1}}},
\end{equation*}
which is sufficient for our purposes as it implies that
\begin{equation*}
p_{0} + p_{1} \leq 1 + \sqrt{2} \cdot 2^{- n / 2}
\end{equation*}
for all strategies of dishonest Alice. Therefore, the protocol is $\varepsilon$-binding with $\varepsilon = 2^{(1 - n) / 2}$ decaying exponentially in $n$ (but note that the decay rate is half of the decay rate against classical adversaries).

The two-round protocol gets mapped onto a non-local game precisely because of the assumption of no communication. More specifically, we require that $\sA_{1}$ outputs the answer outside of the future of $\sA_{2}$ receiving the input and vice versa.
\subsection{A new multi-round protocol based on finite-field arithmetic}
\label{sec:multi-round-protocol}
The protocol presented in Section C.3 implements a bit commitment scheme that is provably secure against quantum adversaries. Unfortunately, the commitment time is limited by $s/c$, where $s$ is the spatial separation between Locations $1$ and $2$ and $c$ is the speed of light. If the two sites are constrained to be on the surface of the Earth then the commitment can only be valid for approximately $42$ milliseconds. Here, we present a new, multi-round scheme which, by adding extra intermediate rounds, allows for an arbitrarily long commitment and we prove its security against classical adversaries. Note, however, that the security guarantee depends on the number of rounds of the protocol (which is proportional to the length of the commitment): the longer the commitment, the more resources (randomness and communication bandwidth) are required to achieve the same level of security.

The protocol consists of $m + 1$ rounds labelled by $k \in [m + 1]$, which obey the causal structure described in Section C.1. The commitment is initiated in the first round ($k = 1$), it is sustained for $k = 2, 3, \ldots, m$ and is eventually opened in the last round ($k = m + 1$). Let us emphasise that we want Alice to become committed as soon as the first round is over and if she were able to decide on the value of her commitment in the second round we would consider that cheating. This is necessary to argue that the commitment really begins in the first round. All the rounds before the open phase ($1 \leq k \leq m$) have the same communication pattern: first $\sB_{i}$ sends an $n$-bit string to $\sA_{i}$ and then she replies with another $n$-bit string. In the last round $\sA_{i}$ sends $\sB_{i}$ a bit (her commitment) and an $n$-bit string (proof of her commitment). We will denote the $n$-bit string announced by Bob (Alice) in the $k\th$ round by $x_{k}$ ($y_{k}$) regardless of whether he/she is honest or not.

Before the protocol begins Alice generates $m$ strings of $n$-bits (private from Bob), denoted by $\{a_{k}\}_{k = 1}^{m}$, drawn independently, uniformly at random from $\{0, 1\}^{n}$, and distributes them to $\sA_{1}$ and $\sA_{2}$. Similarly, Bob generates $m$ strings of $n$-bits (private from Alice), denoted by $\{b_{k}\}_{k = 1}^{m}$, drawn independently, uniformly at random from $\{0, 1\}^{n}$, and distributes them to $\sB_{1}$ and $\sB_{2}$.

The protocol goes as follows
\begin{enumerate}
\item (commit) In the first round $\sB_{1}$ sends $x_{1} = b_{1}$ to $\sA_{1}$ and she replies with $y_{1} = d \cdot x_{1} \oplus a_{1}$. This initiates the commitment.
\item (sustain) In the $k\th$ round (for $2 \leq k \leq m$) $\sB_{i}$ sends $\sA_{i}$ the string $x_{k} = b_{k}$ and she replies with $y_{k} = (x_{k} * a_{k - 1}) \oplus a_{k}$.
\item (open) In the $(m + 1)\th$ round $\sA_{i}$ sends $d$ and $y_{m + 1} = a_{m}$ to $\sB_{i}$.
\item (verify) Bob collects data from $\sB_{1}$ and $\sB_{2}$ and accepts the commitment iff the strings announced by $\sA_{1}$ and $\sA_{2}$ satisfy
\begin{equation}
\begin{aligned}
\label{eq:acceptance-condition}
y_{m + 1} \; = \; y_{m} \; \oplus \; b_{m} * y_{m - 1} \; \oplus \; b_{m} * b_{m - 1} * y_{m - 2} \; \oplus \; \ldots\\
\ldots \; \oplus \; b_{m} * b_{m - 1} * \ldots * b_{2} * y_{1} \; \oplus \; d \cdot b_{m} * b_{m - 1} * \ldots * b_{1}.
\end{aligned}
\end{equation}
\end{enumerate}
Note that the private strings of Alice and Bob as well as the messages they exchange are actually random variables and that is how they must be treated in the analysis.
\subsubsection{Correctness}
It is easy to verify (by induction) that if Alice and Bob follow the protocol then the condition~\eqref{eq:acceptance-condition} is satisfied for any choice of strings $\{a_{k}\}_{k = 1}^{m}$ and $\{b_{k}\}_{k = 1}^{m}$.
\subsubsection{Security for honest Alice}
We start with a lemma which formalises the intuition that if we take an arbitrary random variable taking values in a finite field and perform (finite field) addition with a uniform and uncorrelated random variable then there will be no correlations between the input and the output (or any function thereof). More specifically, in the following lemma $Y$ is a random variable from which the input is generated using function $g$, $X$ is the fresh (finite field) randomness and $h$ is a function allowing us to condition on a certain subset of values of $Y$.
\begin{lem}
\label{lem:random-variables}
Let $\cX = \amsbb{F}_{q}$ and $\cY, \cZ$ be arbitrary finite sets. Let $X$ and $Y$ be two random variables taking values in $\cX$ and $\cY$, respectively, such that $X$ is uniform and independent from $Y$
\begin{equation}
\label{eq:assumption}
\Pr[X = x, Y = y] = q^{-1} \cdot \Pr[Y = y],
\end{equation}
for all $x \in \cX$ and $y \in \cY$. Then for arbitrary functions $g : \cY \to \cX$, $h : \cY \to \cZ$ and arbitrary fixed $x \in \cX$, $z \in \cZ$ it holds that
\begin{equation*}
\Pr[X + g(Y) = x \, | \, h(Y) = z] = q^{-1}.
\end{equation*}
\end{lem}
\begin{proof}
Note that
\begin{align*}
&\Pr[X + g(Y) = x, h(Y) = z] = \sum_{y \in \cY} \Pr[X = x - g(y), h(y) = z , Y = y]\\
&= \sum_{y \in \cY \atop h(y) = z} \Pr[X = x - g(y), Y = y] = \sum_{y \in \cY \atop h(y) = z} q^{-1} \cdot \Pr[Y = y] = q^{-1} \cdot \Pr[h(Y) = z],
\end{align*}
where the second last equality follows from applying the assumption~\eqref{eq:assumption} to every term of the sum.
\end{proof}
\begin{prop}
If Alice is honest then the protocol is perfectly hiding.
\end{prop}
\begin{proof}
As explained in Section~\ref{sec:honest-alice} we need to show that the transcripts for $d = 0$ and $d = 1$ after $m$ rounds (immediately before the open phase) are indistinguishable
\begin{equation*}
\Pr[Y_{1} = y_{1}, Y_{2} = y_{2}, \ldots, Y_{m} = y_{m} | d = 0] = \Pr[Y_{1} = y_{1}, Y_{2} = y_{2}, \ldots, Y_{m} = y_{m} | d = 1]
\end{equation*}
for all $y_{1}, y_{2}, \ldots, y_{m}$. In fact, we will show by induction that
\begin{equation}
\label{eq:transcript-distribution}
\Pr[Y_{1} = y_{1}, Y_{2} = y_{2}, \ldots, Y_{t} = y_{t} | d = b] = 2^{-n t},
\end{equation}
for all $t \in [m]$ regardless of the value of $b \in \{0, 1\}$, which clearly satisfies the indistinguishability condition.

Honest Alice will follow the protocol, which means that $\{A_{k}\}_{k = 1}^{m}$ are drawn independently, uniformly at random from $\{0, 1\}^{n}$, the value of her commitment is $d$ and then Alice's message in $k\th$ round is
\begin{equation}
\label{eq:yk}
Y_{k} =
\begin{cases}
d \cdot X_{1} \oplus A_{1} &\nbox{for} k = 1,\\
Y_{k} = (X_{k} * A_{k - 1}) \oplus A_{k} &\nbox{for} 2 \leq k \leq m.
\end{cases}
\end{equation}
Bob, on the other hand, is only limited by the causal constraints, which means that his message in the $k\th$ round might depend on some randomness preshared between $\sB_{1}$ and $\sB_{2}$, denoted by $R_{B}$, and all the responses of Alice which belong to the past of the $k\th$ round. Therefore, without loss of generality his message in the $k\th$ round is
\begin{equation}
\label{eq:xk}
X_{k} = f_{k}(R_{B}, Y_{1}, Y_{2}, \ldots Y_{k - 2})
\end{equation}
for some arbitrary function $f_{k}$ (we include all randomness used by Bob in $R_{B}$ so $f_{k}$ is deterministic).

In this scenario the full transcript is a deterministic function of Alice's commitment $d$, her private randomness $\{A_{k}\}_{k = 1}^{m}$ and Bob's preshared randomness $R_{B}$. For every string announced by Alice and Bob we can explicitly find the subset of random variables it may depend on as listed in the table below
\begin{center}
%
\setlength{\tabcolsep}{0.6cm}
\begin{tabular}{c c}
\textnormal{message} & \textnormal{random variables it might depend on}\\
\hline
$X_{1}$ & $R_{B}$\\
$X_{2}$ & $R_{B}$\\
$X_{3}$ & $d, R_{B}, A_{1}$\\
\vdots & \vdots\\
$X_{k}$ & $d, R_{B}, A_{1}, A_{2}, \ldots, A_{k - 2}$\\
\vdots & \vdots\\
$X_{m}$ & $d, R_{B}, A_{1}, A_{2}, \ldots, A_{m - 2}$\\
$Y_{1}$ & $d, R_{B}, A_{1}$\\
$Y_{2}$ & $d, R_{B}, A_{1}, A_{2}$\\
$Y_{3}$ & $d, R_{B}, A_{1}, A_{2}, A_{3}$\\
\vdots & \vdots\\
$Y_{k}$ & $d, R_{B}, A_{1}, A_{2}, \ldots, A_{k}$\\
\vdots & \vdots\\
$Y_{m}$ & $d, R_{B}, A_{1}, A_{2}, \ldots, A_{m}$\\
\end{tabular}
\end{center}
First, we verify that Eq.~\eqref{eq:transcript-distribution} holds for $t = 1$
\begin{equation*}
\Pr[Y_{1} = y_{1} | d = b] = \Pr[ b \cdot X_{1} \oplus A_{1} = y_{1}] = \Pr[ b \cdot f_{1}(R_{B}) \oplus A_{1} = y_{1}] = 2^{-n},
\end{equation*}
where the first two equalities follow from Eqs.~\eqref{eq:yk} and \eqref{eq:xk}, respectively. The last equality is a direct consequence of Lemma~\ref{lem:random-variables} (in a simplified form: no conditioning) applied to $X = A_{1}$, $Y = (b, R_{B})$, $g(Y) = b \cdot f_{1}(R_{B})$. Now, suppose that Eq.~\eqref{eq:transcript-distribution} holds for $t = k$. Then
\begin{align*}
\Pr&[Y_{1} = y_{1}, \ldots, Y_{k + 1} = y_{k + 1} | d = b]\\
&= \Pr[Y_{k + 1} = y_{k + 1} | d = b, Y_{1} = y_{1}, \ldots, Y_{k} = y_{k}] \cdot \Pr[Y_{1} = y_{1}, \ldots, Y_{k} = y_{k} | d = b]\\
&= \Pr[(X_{k + 1} * A_{k}) \oplus A_{k + 1} = y_{k + 1} | d = b, Y_{1} = y_{1}, \ldots, Y_{k} = y_{k}] \cdot 2^{-n k}\\
&= 2^{-n} \cdot 2^{-n k} = 2^{- (n + 1) k},
\end{align*}
where the second last inequality follows from applying Lemma~\ref{lem:random-variables} to
\begin{align}
X &= A_{k + 1},\nonumber\\
Y &= (b, R_{B}, A_{1}, \ldots, A_{k}),\nonumber\\
\label{eq:gy}
g(Y) &= X_{k + 1} * A_{k},\\
h(Y) &= (Y_{1}, Y_{2}, \ldots, Y_{k})\nonumber.
\end{align}
Note that it is not immediately obvious and the reader should verify (using the table presented above) that the quantities on the right-hand side of Eq.~\eqref{eq:gy} are functions of $Y$ alone, and therefore satisfy the assumptions of the lemma. This shows that Eq.~\eqref{eq:transcript-distribution} holds for $t = k + 1$ and so by induction it must hold for all $t \in [m]$. Therefore, even just before the open phase the transcript contains no information about Alice's commitment and the protocol is perfectly hiding.
\end{proof}
\subsubsection{Security for honest Bob}
\begin{prop}
If Bob is honest then the protocol is $\varepsilon$-binding for $\varepsilon = \omega_{m}$ defined in Eq.~\eqref{eq:omega-definition}.
\end{prop}
\begin{proof}
Honest Bob will follow the protocol so $X_{k} = B_{k}$, where $\{B_{k}\}_{k = 1}^{m}$ are drawn independently, uniformly at random from $\{0, 1\}^{n}$. Let $R_{A}$ be any randomness preshared by $\sA_{1}$ and $\sA_{2}$ before the protocol begins. The most general cheating strategy for Alice allowed by the security model described in Section~\ref{sec:honest-bob} is a collection of (deterministic) functions, $\{f_{k}\}_{k = 1}^{m + 1}$, all of which output an $n$-bit string while their inputs are as described below.
\begin{itemize}
\item Alice's message in the commit phase might depend on the preshared randomness and the first message of Bob
\begin{equation*}
Y_{1} = f_{1}(R_{A}, B_{1}).
\end{equation*}
\item Alice's messages during the sustain phase ($k \in \{2, 3, \ldots, m\}$) might depend on the preshared randomness, Bob's messages from the past and the bit she is trying to unveil $d$\\
\begin{equation*}
Y_{k} = f_{k}(R_{A}, B_{1}, B_{2}, \ldots, B_{k - 2}, B_{k}, d).
\end{equation*}
Note that $Y_{k}$ must not depend on $B_{k - 1}$ because, by assumption, it does not belong to the past of the $k\th$ round.
\item Alice's message in the open phase might depend on the preshared randomness, Bob's messages from the past and the bit she is trying to unveil $d$\\
\begin{equation*}
Y_{m + 1} = f_{m + 1}(R_{A}, B_{1}, B_{2}, \ldots, B_{m - 1}, d).
\end{equation*}
Again, $Y_{m + 1}$ must not depend on $B_{m}$.
\end{itemize}
The commitment to $d$ will be accepted iff~\eqref{eq:acceptance-condition} is satisfied for that value of $d$ and let us denote this event by $H_{d}$. By definition $p_{d} = \Pr[H_{d}]$ and since both events are defined over $(R_{A}, B_{1}, B_{2}, \ldots, B_{m})$ it is meaningful to talk about $H_{0} \lor H_{1}$ and $H_{0} \wedge H_{1}$. (Note that this reasoning does not work in the quantum setting since $H_{0}$ and $H_{1}$ are \emph{not defined} simultaneously.) To bound $p_{0} + p_{1}$ we use $\Pr[H_{0}] + \Pr[H_{1}] = \Pr[H_{0} \lor H_{1}] + \Pr[H_{0} \wedge H_{1}] \leq 1 + \Pr[H_{0} \wedge H_{1}]$. The event $H_{0} \wedge H_{1}$ happens if~\eqref{eq:acceptance-condition} is satisfied for both values of $d$. Define $K$ to be the event that the XOR of the two conditions (i.e.~ Eq.~\eqref{eq:acceptance-condition} for $d = 0$ and $d = 1$) is satisfied
\begin{align*}
K \iff B_{1} * B_{2} * \ldots * B_{m} &= g_{m + 1}(R_{A}, B_{1}, B_{2}, \ldots, B_{m - 1}) \oplus g_{m}(R_{A}, B_{1}, B_{2}, \ldots, B_{m - 2}, B_{m})\\
&\bigoplus_{k = 2}^{m - 1} B_{m} * B_{m - 1} * \ldots * B_{k + 1} * g_{k}(R_{A}, B_{1}, B_{2}, \ldots, B_{k - 2}, B_{k}),
\end{align*}
where $g_{k}(R_{A}, B_{1}, B_{2}, \ldots, B_{k - 2}, B_{k}) = f_{k}(R_{A}, B_{1}, B_{2}, \ldots, B_{k - 2}, B_{k}, d = 0) \oplus f_{k}(R_{A}, B_{1}, B_{2}, \ldots, B_{k - 2}, B_{k}, d = 1)$. Note that since $H_{0} \wedge H_{1} \implies K$ we have $\Pr[H_{0} \wedge H_{1}] \leq \Pr[K]$.

To bound $\Pr[K]$ note that the right-hand side contains exactly $m$ terms, but each of them depends on $(m - 1)$ $B$'s; none of the terms depends on \emph{all $B$'s simultaneously}. The terms corresponding to $2 \leq k \leq m - 1$ have some internal structure (e.g.~the dependence on $B_{m}$ is \emph{not} arbitrary) but we can relax the problem to the case where the $k\th$ term is an arbitrary function of all the $B$'s except for $B_{k}$ denoted by $h_{k}$. The winning condition for the relaxed game is
\begin{equation*}
B_{1} * B_{2} * \ldots * B_{m} = \bigoplus_{k = 1}^{m} h_{k}( B_{ [m] \setminus \{k\} } )
\end{equation*}
and clearly the winning probability is an upper bound on $\Pr[K]$. In Section B.2 we define the optimal winning probability for this game to be $\omega_{m}$. This concludes the proof since
\begin{equation*}
p_{0} + p_{1} \leq 1 + \Pr[K] \leq 1 + \omega_{m}.
\end{equation*}
\end{proof}
Note that a non-trivial upper bound on $\omega_{m}$ (for an arbitrary $m$) can be obtained using a recursive argument presented in Section B.4.
\end{document}